\documentclass[11pt]{article}

\usepackage{amssymb,amsmath,amsthm}
\usepackage{enumerate}
\usepackage{version}
\usepackage{fullpage}
\usepackage{url}
\usepackage{graphicx}
\usepackage{dsfont}
\usepackage{caption}

\newtheorem{theorem}{Theorem}[section]

\newtheorem{observation}[theorem]{Observation}

\newtheorem{definition}[theorem]{Definition}

\newtheorem{corollary}[theorem]{Corollary}
\newtheorem{claim}[theorem]{Claim}
\newtheorem{lemma}[theorem]{Lemma}
\newtheorem{polyrule}{Rule}[section]

\def\ie{{\em i.e.}~}

\newenvironment{proofclaim}{
	\noindent \emph{Proof.}
}{%
	\hfill $\diamond$ \\
}

\newcommand{\BIT}{{\sc Dense Betweenness}}
\newcommand{\FAST}{\textsc{Feedback Arc Set in Tournaments}}
\newcommand{\rBIT}{{\sc $r$-Dense Betweenness}}
\newcommand{\FASHT}{\textsc{$r$-Dense Feedback Arc Set}}
\newcommand{\rFAST}{\textsc{$r$-Dense Transitive Feedback Arc Set}}

\usepackage{mathrsfs}

\title{Linear vertex-kernels for several dense ranking $r$-CSPs}

\author{Anthony Perez\thanks{LIFO - Universit\'e d'Orl\'eans}}

\date{}

\begin{document}

\maketitle

\begin{abstract}

\noindent A {\sc ranking $r$-constraint satisfaction} problem (ranking $r$-CSP for short) consists of a ground set of vertices $V$, an arity 
$r \geqslant 2$, a parameter $k \in \mathbb{N}$ and a \emph{constraint system} $c$, where $c$ is a function which maps 
rankings of $r$-sized sets $S \subseteq V$ to $\{0,1\}$. The objective is to 
decide if there exists a ranking $\sigma$ of the vertices satisfying all but at most $k$ constraints (\ie $\sum_{S \subseteq V, |S| = r} c(\sigma(S)) \leqslant k$). Famous 
ranking $r$-CSPs include \FAST{} and \BIT{}~\cite{ALS09,KS10}. 
We consider such problems from the kernelization viewpoint.
We prove that so-called \emph{$l_r$-simply characterized} ranking $r$-CSPs admit linear vertex-kernels whenever they admit constant-factor approximation algorithms. This implies that 
\rBIT{} and \rFAST{}~\cite{KS10}, two natural generalizations of the previously mentioned problems, admit linear vertex-kernels. Moreover, we introduce another generalization of \FAST{}, which does not fit the aforementioned framework. Based on techniques from~\cite{CFR06} we obtain a $5$-approximation, and then provide a linear vertex-kernel. 
\end{abstract}

\section*{Introduction}
\label{sec:intro}

Parameterized complexity is a powerful theoretical framework to cope with NP-Hard problems. The aim is to identify some \emph{parameter} $k$, 
independent from the instance size $n$, which captures the exponential growth of the complexity to solve 
the problem at hand. A parameterized problem is said to be \emph{fixed parameter tractable} (FPT for short) whenever it can be solved in $f(k) \cdot n^{O(1)}$ time, 
where $f$ is any computable function~\cite{DF99,Nie06}. In this extended abstract, we focus on \emph{kernelization}, which is one of the most efficient technique to design parameterized algorithms~\cite{Bod09}. A \emph{kernelization algorithm} (or kernel for short) 
for a parameterized problem $\Pi$ is a \emph{polynomial-time} algorithm that given an 
instance $(I, k)$ of $\Pi$ outputs an \emph{equivalent} instance $(I',k')$ of $\Pi$ 
such that $|I'| \leqslant g(k)$ and $k' \leqslant k$. The function $g$ is said to be 
the \emph{size} of the kernel, and $\Pi$ admits a \emph{polynomial kernel} 
whenever $g$ is a polynomial. A well-known result states that a (decidable) parameterized problem 
is FPT if and only if it admits a kernel~\cite{Nie06}. Observe that this theoretical result provides kernels of \emph{super-polynomial} size. 
Recently, several results gave evidence that 
there exist parameterized problems that \emph{do not} admit \emph{polynomial} kernels~(under some complexity-theoretic assumption, see \emph{e.g.}~\cite{BDFH09,BJK11}). Hence, determining the existence of \emph{polynomial} kernels is of important interest from both 
the theoretical and practical viewpoints.  \\

We mainly study 
ranking $r$-CSPs from the kernelization viewpoint. 
A ranking $r$-CSP consists of a ground set of vertices $V$, an arity 
$r \geqslant 2$ and a \emph{constraint system} $c$, where $c$ is a function which maps rankings (\ie orderings) of $r$-sized sets $S \subseteq V$ to $\{0,1\}$.~The objective is to find a ranking  $\sigma$ of the vertices 
that minimizes the number of unsatisfied constraints (\ie $\sum_{S \subseteq V, |S| = r} c(\sigma(S))$). We study the decision version of these problems, where the instance comes together with some parameter $k \in \mathbb{N}$ and the aim is to decide if there exists a ranking of the vertices satisfying \emph{all 
but at most $k$} constraints.~We focus on problems where a constraint $S$ 
can be represented by a set of \emph{selected vertices}, which determines the conditions that 
a ranking must verify in order to satisfy $S$. Moreover, we consider such problems 
on \emph{dense instances}, where \emph{every} set of $r$ vertices is a constraint. 
A lot of well-studied problems can be expressed in such terms~\cite{ALS09, KS11}. 
For instance, \FAST{} fits this 
framework with $r = 2$, any arc $uv$ being satisfied by a ranking $\sigma$ (\ie $c(\sigma(uv)) = 0$) iff $u <_\sigma v$. 
Such problems can be equivalenty stated in terms of  
modification problems: can we \emph{edit} at most $k$ constraints to obtain an instance that admits a ranking satisfying all its constraints? \\


\noindent \textbf{Related results.} While a lot of kernelization results are known for \emph{graph} modification 
problems~\cite{BP11,KW09,Tho10,vBMN10}, fewer results exist regarding directed graph and hypergraph modification problems. 
An example of polynomial kernel for a directed graph modification problem is the 
quadratic vertex-kernel for \textsc{Transitivity Editing}~\cite{WKNU09}.~Regarding dense ranking $r$-CSPs, \FAST{} and \BIT{} are NP-Complete~\cite{AA07,Alo06,CTY07} but fixed parameter tractable \cite{ALS09,KS10}, and both admit a linear vertex-kernel~\cite{PPT11}. 
Concerning ranking $r$-CSPs, Karpinski and Schudy~\cite{KS11} recently showed PTASs and subexponential parameterized algorithms for so-called (\emph{weakly})-fragile ranking $r$-CSPs. A constraint is \emph{fragile} (resp. \emph{weakly-fragile}) if whenever it is satisfied by one ranking then 
making one single move (resp. making one of the following moves: swapping the first two vertices, the last two vertices 
or making a cyclic move) makes it unsatisfied. These problems contain in particular \FAST{} and \BIT{}. \\

\noindent \textbf{Our results.} Following this line of research, we provide several linear  vertex-kernels for dense ranking $r$-CSPs. 
More precisely, we introduce a particular type of 
ranking $r$-CSPs, called \emph{$l_r$-simply characterized}, and prove that such problems admit linear vertex-kernels whenever they admit constant-factor approximation algorithms (Section~\ref{sec:general}). Surprisingly, our kernels mainly 
use a modification of the classical sunflower reduction rule, which usually provides polynomial kernels~\cite{ALS09,BP11,GHPP12}. This result improves the size of the kernel for \BIT{}~\cite{PPT11}. Moreover, it implies linear vertex-kernels for \rBIT{} and \rFAST{}, two natural generalizations of \FAST{} and \BIT{}. Both problems were left open by Karpinski and Schudy~\cite{KS11}. 
Finally, we introduce a different generalization of \FAST{}, which allows more freedom on the satisfiability of a given constraint. 
Based on ideas used for \FAST{}~\cite{CFR06}, we prove that this problem admits a $5$-approximation algorithm, and then obtain a linear vertex-kernel (Section~\ref{subsec:rfasht}). 

\section{Preliminaries}
\label{sec:prelim}

A ranking $r$-CSP consists of a ground set of vertices $V$, an arity $r \geqslant 2$, a parameter $k \in \mathbb{N}$ and a \emph{constraint system} $c$, where $c$ is a function which maps rankings (\ie orderings) of $r$-sized sets $S \subseteq V$ to $\{0,1\}$. In a slight abuse of notation, we refer to a set of vertices $S \subseteq V$, $|S| = r$, as a \emph{constraint} (when we are actually referring to $c$ applied to rankings of $S$). A constraint $S$ is \emph{non-trivial} whenever there exists a ranking $\sigma$ such that $c(\sigma(S)) = 1$. In the following, we always mean non-trivial constraints when speaking of constraints. A constraint $S$ is \emph{satisfied} by a ranking $\sigma$ whenever $c(\sigma(S)) = 0$, in which case $S$ is said to be \emph{consistent} w.r.t. $\sigma$ (we forget the mention \emph{w.r.t. $\sigma$} whenever the context is clear). Otherwise, we say that $S$ is inconsistent. 
Similarly, a ranking $\sigma$ is \emph{consistent} with the constraint system $c$ if it does not contain any inconsistent constraint, and \emph{inconsistent} otherwise. The objective of a ranking $r$-CSP is to find a ranking of the vertices 
with \emph{at most $k$} inconsistent constraints. 
We consider ranking $r$-CSPs where a constraint $S$ can 
be represented by a subset $sel(S) \subseteq S$ of \emph{selected vertices}, that determines  the conditions that a ranking must verify in order to satisfy $S$. 
An equivalent formulation of these problems is the 
following: is it possible to \emph{edit} at most $k$ constraints so that there exists 
a ranking consistent with the new constraint system? By \emph{editing a constraint}, 
we mean that we modify its set of selected vertices. \\

We consider \emph{dense} instances, where \emph{every} subset of $r$ vertices of $V$ is a constraint. Let $R = (V,c)$ be an instance of any ranking $r$-CSP. Given a set of vertices $V' \subseteq V$, we define the instance \emph{induced by $V'$} (and denote it $R[V']$) as the constraint system $c$ restricted to $r$-sized subsets of $V'$. A set of vertices $C \subseteq V$ is a \emph{conflict} if there does not exist any ranking consistent with the instance induced by $C$.  
We mainly study the following problems.\\

\fbox{\begin{minipage}{0.9\textwidth}
	\rBIT{}:\\
	\textbf{Input}: A set of vertices $V$, a constraint system $c$, where a constraint $S = \{s_1, \ldots, s_r\}$ contains two \emph{selected vertices} $s_i$ and $s_j$, $1 \leqslant i < j \leqslant r$, and is satisfied by a ranking $\sigma$ iff $s_i <_\sigma s_l <_\sigma s_j$ or $s_j <_\sigma s_l <_\sigma s_i$ holds for $1 \leqslant l \leqslant r$, $l \neq \{i,j\}$. \\
	\textbf{Parameter}: $k$. \\
	\textbf{Output}: A ranking $\sigma$ of $V$ that satisfies all but at most $k$ constraints. 
\end{minipage}}	
\\[0.5cm]

\fbox{\begin{minipage}{0.9\textwidth}
	\FASHT{}:\\ 
	\textbf{Input}: A set of vertices $V$, a constraint system $c$, where a constraint $S$ contains one selected vertex $s$ and is satisfied 
by a ranking $\sigma$ if $u <_\sigma s$ for any $u \in S \setminus \{s\}$. \\
	\textbf{Parameter}: $k$. \\
	\textbf{Output}: A ranking $\sigma$ of $V$ that satisfies all but at most $k$ constraints. 
\end{minipage}}
\\[0.3cm]

We also consider another generalization of the \FAST{} problem, namely \rFAST{} ({\sc $r$-FAST})~\cite{KS11}, where a constraint $S$ corresponds to an acyclic tournament and is satisfied by a ranking $\sigma$ if 
	$\sigma$ is the transitive ranking of the corresponding tournament.

\begin{figure}[h]
\label{fig:constraints}

	\centerline{\includegraphics[scale=3]{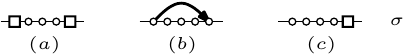}}
	\caption{Illustration of $(a)$ {\sc $r$-BIT}, $(b)$ {\sc $r$-TFAST} and $(c)$ {\sc $r$-FAST} for $r = 5$. The square vertices represent the 
	selected vertices of the constraints, that are consistent with $\sigma$.}

\end{figure} 

\newpage

\noindent \textbf{Ordered instances.} In the following, we consider instances whose vertices are ordered under some fixed ranking $\sigma$ (\ie instances of the form $R_\sigma = (V, c, \sigma)$). Given any constraint $S = \{s_1, \ldots, s_r\}$, with $s_i <_\sigma s_{i+1}$ for $1 \leqslant i < r$, $span(S)$ denotes the set of vertices $\{v \in V : s_1 \leqslant_\sigma v \leqslant_\sigma s_r\}$. A constraint $S$ is \emph{unconsecutive} if 
$|span(S)| > r$, and \emph{consecutive} otherwise. Given $V' \subseteq V$, 
$R_\sigma[V']$ denotes the instance $R[V']$ ordered under $\sigma$. Finally, given a ranking $\sigma$ over $V$ and an inconsistent constraint $S$, we say that we \emph{edit $S$ w.r.t. $\sigma$} whenever we edit its selected vertices so that it becomes 
satisfied by $\sigma$.

\section{Simple characterization and sunflower}
\label{sec:general}

In this section, we describe the general framework of our kernelization algorithms, using a modification of the \emph{sunflower rule} together with the notion of \emph{simple characterization}. We first define the notion of sunflower, which has been widely used to obtain polynomial kernels for modification problems~\cite{Abu10,ALS09,BP11,GHPP12}.
An \emph{edition} is a set of constraints $\mathcal{F}$ such that one can obtain a consistent instance by editing constraints in $\mathcal{F}$.  

\begin{definition}
	A \emph{sunflower} $\mathcal{S}$ is a set of conflicts $\{C_1, \ldots, C_m\}$ pairwise intersecting in \emph{exactly one} constraint $S$, called the \emph{center} of $\mathcal{S}$.
\end{definition} 

\begin{lemma}
\label{lem:flower}
	Let $R = (V, c)$ be an instance of any ranking $r$-CSP, and $S$ be the center of a sunflower $\mathcal{S} = \{C_1, \ldots, C_m\},$ $m > k$. Any edition of size at most $k$ has to edit $S$.
\end{lemma}

\begin{proof}
	Let $\mathcal{F}$ be any edition of size at most $k$, and assume that $\mathcal{F}$ does not edit $S$. This means that $\mathcal{F}$ must contain one constraint for every conflict $C_i$, $1 \leqslant i \leqslant m$. Since $m > k$, we conclude that $\mathcal{F}$ contains more than $k$ constraints, a contradiction.
 \end{proof}

Observe that the sunflower rule cannot be applied directly on ranking $r$-CSPs, $r \geqslant 3$, since it may be the case that there exist several ways to edit the center of a given sunflower. In order to deal with this issue, we introduce the notion of \emph{simple characterization} for ranking $r$-CSPs. 
Roughly speaking, a ranking $r$-CSP is \emph{$l_r$-simply characterized} if for any ordered instance, any set of $l_r$ vertices which involve \emph{exactly one} inconsistent constraint is a conflict. We first give a formal definition of this notion, and next describe how the sunflower rule can be modified for such problems. 

\begin{definition}[Simple characterization]
\label{def:simplecharac}
	Let $\Pi$ be a ranking $r$-CSP, $R_\sigma 
	= (V, c, \sigma)$ be any ordered instance of $\Pi$, and $l_r \in \mathbb{N}$. The ranking $r$-CSP $\Pi$ is \emph{$l_r$-simply characterized} iff any set $C \subseteq V$ of $l_r$ vertices such that $R_\sigma[C]$ contains exactly one inconsistent constraint is a conflict. 
	
\end{definition}


\begin{definition}[Simple sunflower]
\label{def:simplesunflower}
	 Let $R_\sigma = (V, c, \sigma)$ be an ordered instance of a $l_r$-simply 
	 characterized ranking $r$-CSP. A sunflower $\mathcal{S} = \{C_1, \ldots, C_m\}$ of 
	 $R_\sigma$ is \emph{simple} if 
	 its center is the only inconsistent constraint in $R_\sigma[C_i]$, $1 \leqslant i \leqslant m$.
	 
\end{definition}

\begin{polyrule}
\label{rule:correct}
	Let $\Pi$ be a $l_r$-simply  characterized ranking $r$-CSP. Let $R_\sigma = (V, c, \sigma)$ be an ordered instance of $\Pi$ 
	and $\mathcal{S} = \{C_1, \ldots, C_m\}$, $m > k$, be a simple sunflower of center $S$. 
	Edit $S$ w.r.t. $\sigma$ and decrease $k$ by $1$.
\end{polyrule}

\begin{lemma}
\label{lem:correct}
	Rule~\ref{rule:correct} is sound. 

\end{lemma}

\begin{proof}
	Let $\mathcal{F}$ be any edition of size at most $k$: by Lemma~\ref{lem:flower}, $\mathcal{F}$ must contain $S$. Since $|\mathcal{F}| \leqslant k$ and $m > k$, there exists $1 \leqslant i \leqslant m$ such that $S$ is the only constraint edited by $\mathcal{F}$ in $R[C_i]$. Assume that $S$ was not edited w.r.t. $\sigma$: since no other constraint has been edited in $R[C_{i}]$, $R_\sigma[C_i]$ still contains exactly one inconsistent constraint (namely $S$). Since $\Pi$ is $l_r$-simply characterized, it follows that $C_{i}$ defines a conflict, contradicting the fact that $\mathcal{F}$ is an edition.
 \end{proof}

The main problem that remains is to compute such a sunflower in polynomial time. The following result will allow us to do so, providing that $V$ contains sufficiently many vertices (w.r.t. parameter $k$). 

\begin{lemma}
	\label{lem:correctSC}
	Let $\Pi$ be a $l_r$-simply characterized ranking $r$-CSP, and $R_\sigma = (V, c, \sigma)$ be an ordered instance of $\Pi$ with at most $p \geqslant 1$ inconsistent constraints. If $|V| > p(l_r - r) + (l_r - r) \cdot (k+1) + r$, there exists a simple sunflower $\{C_1, \ldots, C_m\}$, $m > k$, that can be found in polynomial time.
\end{lemma}

\begin{proof}
	Let $S$ be any inconsistent constraint of $R_\sigma$. Since $R_\sigma$ contains at 
	most $p$ inconsistent constraints, there are at most $p$ disjoint sets $P_i$, $1 \leqslant i \leqslant p$, such that $|P_i| = l_r - r$ and 
	$R_\sigma[S \cup P_i]$ contains more than one inconsistent constraint. 
	It follows that there exist at least $m > k$ disjoint sets $\{S_1, \ldots, S_m\}$ of size $l_r - r$ such that: $(i)$ $C_i = S \cup S_i$ contains $l_r$ vertices and $(ii)$ $R_\sigma[C_i]$ contains \emph{exactly one} inconsistent constraint, $1 \leqslant i \leqslant m$. Since $\Pi$ is $l_r$-simply  characterized,  $C_i$ defines a conflict for every $1 \leqslant i \leqslant m$. It follows that $\{C_1, \ldots, C_m\}$ is a simple sunflower of center $S$.
 \end{proof}

In order to obtain the ranking necessary to apply Lemma~\ref{lem:correctSC}, we rely on the existence of a constant-factor approximation algorithms for the problem at hand. 

\begin{theorem}
\label{thm:kernelSC}
	Let $\Pi$ be a $l_r$-simply  characterized ranking $r$-CSP that admits a $q$-factor approximation algorithm for some constant $q > 0$. Then $\Pi$ admits a kernel with at most $k[(q+1) \cdot (l_r - r)] + l_r$ vertices.
\end{theorem}

\begin{proof}
	Let $R = (V, c)$ be an instance of $\Pi$. We start by computing a ranking $\sigma$ containing $p$ inconsistent constraints using the $q$-factor approximation algorithm. Observe that we can assume that $p > k$, since otherwise we simply return a small trivial \textsc{Yes}-instance. Similarly, we can assume that $p \leqslant qk$, since otherwise we return a small trivial \textsc{No}-instance. We now consider $R_\sigma = (V, c, \sigma)$ and assume that $|V| > p(l_r - r) + (l_r - r) \cdot (k+1) + r$: by Lemma~\ref{lem:correctSC}, it follows that there exists a simple sunflower that can be found in polynomial time, and hence Rule~\ref{rule:correct} can be applied. Since conditions of Lemma~\ref{lem:correctSC} still hold after an application of Rule~\ref{rule:correct}, repeating this process on $R_\sigma$ implies that every inconsistent constraint must be edited. Since $p > k$, we return a small trivial \textsc{No}-instance in such a case. This means that $|V| \leqslant qk(l_r - r) + (l_r - r) \cdot (k+1) + r$, implying the result.
 \end{proof}

\section{Simple characterization of several ranking $r$-CSPs}
\label{sec:ranking}

\subsection{{\sc \BIT{}}} 

As a first consequence of Theorem~\ref{thm:kernelSC}, we improve the size of the linear vertex-kernel for 
BIT from $5k$~\cite{PPT11} to $(2 + \epsilon)k + 4$ for any $\epsilon > 0$. The result directly follows 
from the fact that BIT admits a PTAS~\cite{KS11} and is $4$-simply characterized~\cite{PPT11}. 

\begin{corollary}
\label{thm:kernelbit}
	\BIT{} admits a kernel with at most $(2 + \epsilon)k + 4$ vertices.
\end{corollary}

\subsection{{\sc \rBIT{}} ($r \geqslant 4$)} 
\label{subsec:rbit}

We now consider the {\sc $r$-BIT} problem with constraints of arity $r \geqslant 4$. The main difference with the case $r = 3$ lies in the fact that there is no longer a \emph{unique} way to rank the vertices of a consistent instance in order to satisfy all constraints. In particular, this means that the problem is not 
$(r+1)$-simply characterized. 
However, as we shall see in Lemma~\ref{lem:rBITdoubleconflict}, {\sc $r$-BIT} is $2r$-simply characterized. To see this, we first need the following result. 

\begin{lemma}
\label{claim:rBITconflict}
	Let $R_\sigma = (V, c, \sigma)$ be an ordered instance of $r$-BIT, and $C = \{s_1, \ldots, s_{r+1}\}$ be a set of $r + 1$ vertices such that $s_i <_\sigma s_{i+1}$, $1 \leqslant i \leqslant r$. Assume that $R_\sigma[C]$ contains exactly one inconsistent constraint $S$. Then $C$ is a conflict if and only if:
	\begin{enumerate}[(i)]
		 \item $S$ is unconsecutive or, 
		 \item  $S = \{s_1, \ldots, s_r\}$ and $sel(S) \neq \{s_1, s_l\}$ with $2 < l < r$ (resp. $S = \{s_2, \ldots, s_{r+1}\}$ and $sel(S) \neq \{s_l, s_{r+1}\}$ with $2 < l < r$). 
	\end{enumerate}
\end{lemma}

\begin{proof}
	Assume first that the vertices of $S$ are unconsecutive (\ie there exists $1 < i < r+1$ such that $S$ is equal to $C \setminus \{s_i\}$). 
	
	\begin{itemize}
	
		\item \emph{Case 1.} Assume that a single move on a selected vertex is sufficient to make $S$ consistent. To be more precise, we assume that we have $sel(S) = \{s_l, s_{r+1}\}$ for some $2 \leqslant l \leqslant r$, $l \neq i$ (the case $sel(S) = \{s_1, \ldots, s_l\}$, 
		$2 \leqslant l \leqslant r$ and $l \neq i$ being similar). Since all constraints induced by $C$ but $S$ are consistent w.r.t. $\sigma$, we know that $sel(C \setminus \{s_j\}) = \{s_1, s_{r+1}\},$ with $2 \leqslant j \leqslant r,$ $ j \neq \{i,l\}$ (which is well-defined since $r \geqslant 4$).

In order to satisfy $S$, any consistent ranking $\rho$ must rank $s_1$ between $s_l$ and $s_{r+1}$, while $C \setminus \{s_j\}$ implies that $s_l$ must be between $s_1$ and $s_{r+1}$. Since no ranking can satisfy both these properties, it follows that there does not exist any ranking consistent with $C$.

	\item \emph{Case 2.} Next, assume that the two selected vertices must be moved to obtain a consistent ranking for $S$. In other words, we have $sel(S) = \{s_l, s_{l'}\}$ for some $2 \leqslant l < l' \leqslant r$, $\{l,l'\} \neq i$. Since all constraints but $S$ are consistent, we know that $sel(C \setminus \{s_{r+1}\}) = \{s_1, s_r\}$. In 
	order to satisfy $S$, any 
	consistent ranking must rank $s_1$ between $s_l$ and $s_{l'}$, while 
	$C \setminus \{s_{r+1}\}$ implies that $s_1$ must be before (or after) both 
	$s_l$ and $s_{l'}$, which cannot be. 

\end{itemize}

	Assume now that the vertices of the inconsistent constraint $S$ are $\{s_1, \ldots, s_r\}$. Moreover, assume that a single move on a selected vertex is sufficient to make $S$ consistent (see Figure~\ref{fig:two} for an illustration). 
	
	\begin{itemize}
	
	\item \emph{Case 1.} If we have $sel(S) = \{s_1, s_l\}$ with $2 < l < r$, then swapping $s_l$ and $s_r$ will yield a consistent ranking. Indeed, observe that since all constraints but $S$ are consistent, it follows that the only constraint where $s_l$ is a selected vertex is $S$, and that $(s_1, s_{r+1})$ and $(s_2, s_{r+1})$ are the selected vertices of the other constraints. Hence swapping $s_l$ and $s_r$ will not modify the consistency of any constraint but $S$. This is the only case where we can get a consistent ranking. 
	
	\item \emph{Case 2.} Assume now that $sel(S) = \{s_1, s_l\}$ with $l = 2$. Then we have $sel(C \setminus \{s_1\}) = \{s_l, s_{r+1}\}$ and $sel(C \setminus \{s_l\}) = \{s_1, s_{r+1}\}$. Observe here that in order to satisfy $S$ and $C \setminus \{s_1\}$, any ranking $\rho$ must rank $s_r$ between $s_1$ and $s_l$ and between $s_l$ and $s_{r+1}$. W.l.o.g., this means 
	that we must have $\{s_1, s_{r+1}\} <_\rho s_r <_\rho s_l$, which is inconsistent 
	with $C \setminus \{s_l\}$. It follows that $C$ is a conflict. 
	
	\end{itemize}
	
	\begin{figure}
	\label{fig:two}
	
		\centerline{\includegraphics[scale=3]{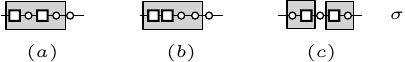}}
		\caption{An illustration of some cases for $r = 4$. The gray boxes represent the inconsistent constraint $S$. $(a)$ $sel(S) = \{s_1, s_l\}$ with $2 < l < r$; $(b)$ $sel(S) = \{s_1, s_l\}$ with $l = 2$ and $(c)$ 
		$S$ is unconsecutive. In the former case $C$ is not a conflict, while $C$ is a conflict in the latter ones. }
	
	\end{figure}
	
	\begin{itemize}

	\item \emph{Case 3.} Assume next that we have $sel(S) = \{s_l, s_r\}$, $1 < l < r$, and $sel(C \setminus \{s_l\}) = \{s_1, s_{r+1}\}$. We know that $sel(C \setminus \{s_j\}) = \{s_1, s_{r+1}\}$, with $2 \leqslant j < r,\ j \neq l$. Once again, this means that any consistent ranking $\rho$ 
	must rank $s_1$ between $s_r$ and $s_l$ in order to satisfy $S$, and $s_l$ between $s_1$ and 
	$s_{r+1}$ in order to satisfy $C \setminus \{s_j\}$) (recall that $l < r$). W.l.o.g., this means that we must have 
	$s_r <_\rho s_1 <_\rho s_l <_\rho s_{r+1}$, which is inconsistent with $C \setminus \{s_l\}$. Hence $C$ is a conflict. 
	
	\item \emph{Case 4.} Finally, assume that the two selected vertices must be moved to obtain a consistent ranking for $S$. In other words, we have $sel(S) = \{s_l, s_{l'}\}$ with $1 < l < l' < r$. We have $sel(C \setminus \{s_r\}) = \{s_1, s_{r+1}\}$. This implies that any consistent ranking must rank $s_1$ between 
$s_l$ and $s_{l'}$ (to satisfy $S$) and $s_{l}$ and $s_{l'}$ between $s_1$ and 
$s_{r+1}$ (to satisfy $C \setminus \{s_r\}$). Since no ranking can satisfy both these 
properties, it follows that $C$ is a conflict. 
	\end{itemize}
	
	The cases where the vertices of the inconsistent constraint are $\{s_2, \ldots, s_{r+1}\}$ is similar to the previous one, the only case yielding a consistent ranking being when $sel(S) = \{s_l, s_{r+1}\}$ and $2 \leqslant l < r$.
\end{proof}

\noindent \textbf{Compatible constraints.}~ 
Given an ordered instance $R_\sigma = (V, c, \sigma)$  of $r$-BIT, we say that an inconsistent constraint $S = \{s_1, \ldots, s_r\}$, $s_1 <_\sigma \ldots <_\sigma s_r$, is \emph{right-} (resp. \emph{left-}) \emph{compatible} whenever $sel(S) = \{s_1, s_l\}$ with $2 < l < r$ (resp. $sel(S) = \{s_l, s_r\}$ with $1 < l < r - 1$), and \emph{right-} (resp. \emph{left-})\emph{incompatible} otherwise. The intuition behind this notion is the following: for any vertex $u$ lying before (resp. after) $S$ in $\sigma$ such that $S$ is the only inconsistent constraint in $R_\sigma[S \cup \{u\}]$, the set $S \cup \{u\}$ does not define a conflict (see Figure~\ref{fig:compatible}). The following result directly follows by definition of a compatible constraint. 

\begin{observation}
\label{obs:comp}
	Any right- (resp. left-)compatible constraint is left- (resp. right-)incompatible. 

\end{observation}

\begin{figure}[t]

	\centerline{\includegraphics[scale=2.25]{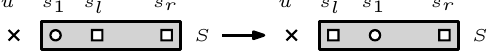}}
	\caption{Illustration of the notion of left-compatible constraints (only $S$ is inconsistent). By definition, 
	$s_1$ is not selected in any constraint in $R_\sigma[S \cup \{u\}]$, and $u$ and $s_{r}$ are selected in every constraint but $S$. Hence swapping $s_1$ and $s_l$ yields a consistent ranking. \label{fig:compatible}}
\end{figure}

\begin{lemma}
\label{lem:rBITdoubleconflict}
	The {\sc $r$-BIT} problem is $2r$-simply characterized.
\end{lemma}

\begin{proof} 	Observe that compatible constraints correspond to the cases of 
Lemma~\ref{claim:rBITconflict} that fail to define a conflict. 
Let $R_\sigma = (V, c, \sigma)$ be an ordered instance of {\sc $r$-BIT} and $C = \{s_1, \ldots, s_{2r}\}$ be a set of $2r$ vertices such that 
	$s_i <_\sigma s_{i+1}$ for $1 \leqslant i < 2r$. Assume that $R_\sigma[C]$ contains exactly 
	one inconsistent constraint $S$. We need to prove that $C$ is a conflict. By Lemma~\ref{claim:rBITconflict}, the result holds if $S$ is neither right nor left-compatible. So we assume w.l.o.g. that $S$ is right-compatible. By Lemma~\ref{claim:rBITconflict} we can also assume that the vertices of $S$ are consecutive and are the first of the ranking, since  otherwise $C$ is a conflict by Lemma~\ref{claim:rBITconflict} and we are done (recall that $S$ is left-incompatible by Observation~\ref{obs:comp}). Hence we may assume that $S = \{s_1, \ldots, s_r\}$ and $sel(S) = \{s_1, s_l\}$ for $2 < l < r$. Moreover, the constraints $S_2 = \{s_l, \ldots, s_r, \ldots, s_{l+r}\}$ and $S_3 = \{s_1, \ldots, s_l, s_r, \ldots, s_{l+r}\}$ (with $|S_3| = r$) have as selected vertices $sel(S_2) = \{s_l, s_{l + r}\}$ and $sel(S_3) = \{s_1, s_{l + r}\}$. 
In order to be consistent with $S$ and $S_2$, any ranking $\rho$ must rank $s_r$ 
between $\{s_1,s_{l+r}\}$ and $s_l$, which is inconsistent with the last constraint (which forces $s_r$ to be between $s_1$ 
and $s_{l+r}$). 
\end{proof}
\vspace{0.2cm}

\begin{corollary}
\label{thm:kernelrBIT}
	$r$-BIT admits a kernel with at most $(2 + \epsilon)rk + 2r$ vertices.
\end{corollary}

\subsection{{\sc \rFAST{}}} 
\label{subsec:rfast}

Karpinski and Schudy~\cite{KS10} considered a particular generalization of the \FAST{} problem, where every constraint $S$ is satisfied by a \emph{one particular} ranking $\sigma$ and no other (\ie $S$ corresponds to an acyclic tournament). 
We show that the {\sc $r$-TFAST}  problem admits a linear vertex-kernel as a particular case of \emph{fragile} ranking $r$-CSP~\cite{KS11}. We say that a ranking 
$r$-CSP is \emph{strongly fragile} whenever a constraint is satisfied by \emph{one particular ranking} and no other. 

\begin{lemma}
\label{lem:lcrfast}
	Let $\Pi$ be any strongly fragile ranking $r$-CSP, $r \geqslant 3$. Then $\Pi$ is $(r+1)$-simply  
	characterized. 
\end{lemma}

\begin{proof}
	Let $R_\sigma = (V, c, \sigma)$ be an ordered instance of $\Pi$, and $C$ be a 
	set of $r + 1$ vertices such that $R_\sigma[C]$ contains exactly one inconsistent constraint $S$. We need to prove that $C$ is a conflict. Assume for a contradiction 
	that this is not the case, \ie that there exists a ranking $\rho$ consistent 
	with $C$. In particular, there exist two vertices $u,v \in S$ such that 
	$u <_\sigma v$ and $v <_\rho u$. Let $S' \neq S$ be any constraint of $R_\sigma[C]$ 
	such that $\{u,v\} 
	\subset S'$ (observe that $S'$ is well-defined since $r \geqslant 3$). 
	Since $S'$ was consistent in $\sigma$ and since $\Pi$ 
	is strongly fragile, $S'$ is inconsistent in $\rho$: a contradiction.
 \end{proof}

\begin{corollary}
	Any strongly fragile ranking $r$-CSP admits a kernel with at most $(2 + \epsilon)k + (r + 1)$  
	vertices. 
\end{corollary}

\subsection{\sc{$r$-Dense Feedback Arc Set} ($r \geqslant 3$)}
\label{subsec:rfasht}

As mentioned previously, the {\sc $r$-TFAST}  problem deals with constraints that are given by a transitive tournament and are thus satisfied by \emph{one particular ranking} and no other. To allow more freedom on the satisfiability of a constraint, we consider a different generalization of this problem, namely {\sc $r$-FAST}. Recall that, in this problem, any constraint $S$ contains a selected vertex $s$ and is satisfied by a ranking $\sigma$ if $u <_\sigma s$ for any $u \in S \setminus \{s\}$. Observe that {\sc $r$-FAST} is not (weakly-)fragile, since swapping the first two vertices of any consistent ranking yields a consistent ranking. Hence, we cannot directly apply the PTASs from~\cite{KS11}. 
\noindent However, the results needed to obtain a $5$-approximation for \textsc{Feedback Arc Set in Tournaments}~\cite{CFR06} can be generalized to the \FASHT{} problem.

\subsubsection{Approximation algorithm}

\begin{definition}
\label{def:indegree}
	Let $R = (V, c)$ be any instance of {\sc $r$-FAST}, and $v \in V$. The \emph{in-degree} of $v$ is the number of constraints where $v$ is selected.
\end{definition}

\noindent \textbf{Algorithm [{\sc Inc-Degree}]}~ Order the vertices of $R$ according to their increasing in-degrees.

\begin{theorem}
\label{thm:approx}
\textsc{Inc-Degree} is a $5$-approximation for {\sc $r$-FAST}.
\end{theorem}

We prove Theorem~\ref{thm:approx} by proving a series of Lemmata. 
For the sake of simplicity, we let $V = \{1, \ldots, n\}$ in the remaining of this Section. For any vertex $v$ of $V$, we call \emph{left constraint} (resp. \emph{above constraint}) any constraint containing $v$ and vertices before (resp. after or both before and after) $v$, and let $L_\sigma(v)$ (resp. $A_\sigma(v)$) be the set of left constraints (resp. above constraints) of $v$. Moreover, we set $l_\sigma(v) = |L_\sigma(v)|$. Finally, given an ordered instance 
$R_\sigma = (V, c, \sigma)$, we define $\mathcal{B}_\sigma$ as the set of inconsistent constraints of $R_\sigma$, and let $b_\sigma = |\mathcal{B}_\sigma|$. 
We need to define a distance function between two rankings $\rho$ and $\gamma$:

\begin{eqnarray*}
	\huge \mathcal{K}(\rho,\gamma) = \displaystyle\sum_{S \subseteq V, |S| = r} \mathds{1}_{(c(\rho(S)) = 1 \land c(\gamma(S)) = 0) \lor (c(\rho(S)) = 0) \land c(\gamma(S)) = 1)}
\end{eqnarray*}

where $\mathds{1}$ denotes the indicative function. This distance gives the number of constraints which are consistent in exactly one out of the two rankings, and thus generalizes the Kendall-Tau distance between two rankings~\cite{CFR06}.

\begin{lemma}
\label{lem:approxone}
Let $\rho : V \rightarrow V$ be any ranking. The following holds:
$$2 \cdot b_\rho \geqslant \displaystyle\sum_{v \in V} |l_\rho(v) - In(v)|$$
\end{lemma}

\begin{proof}
Let $v \in V$ be any vertex. We set: 
\[
	\left\{
		\begin{array}{ll}
			W^-_L(v) = |\{S \in L_\rho(v) : v = sel(S)\}| \\
			W^+_L(v) = |\{S \in L_\rho(v) : v \neq sel(S)\}| \\
			W^-_R(v) = |\{S \in A_\rho(v) : v = sel(S)\}| \\
		\end{array}
	\right.
\]

These numbers respectively represent the left constraints where $v$ is the selected vertex, those where $v$ is not the selected vertex and finally the above constraints where $v$ is the selected vertex. Observe that in the last two cases, the considered constraints are inconsistent. 
By definition, we have that $W^-_L(v) + W^-_R(v) = In(v)$. Moreover, we have $ l_\rho(v) = W^+_L(v) + W^-_L(v)$. Now, observe that:
$$ 2 \cdot b_\rho = \displaystyle\sum_{v \in V} ( W^+_L(v) + W^-_R(v) )$$
To see this, observe that any inconsistent constraint $S = \{s_1, \ldots, s_r\}$ with $s_i <_\rho s_{i+1}$, $1 \leqslant i < r$, will be counted exactly twice in the sum: $(i)$ $S$ is counted in $W^-_R(sel(S))$ and $(ii)$ in $W^+_L(s_r)$, and these are the only vertices for which $S$ will be taken into account.

We conclude the proof by showing that $W^+_L(v) + W^-_R(v) \geqslant |l_\rho(v) - In(v) |$. By the previous observations, we have :
\begin{eqnarray*}
	W^+_L(v) + W^-_R(v) & = & l_\rho(v) + In(v) - 2W^-_L(v) \\
			    & = & |l_\rho(v) - In(v)| + 2(min\{l_\rho(v), In(v)\} - W^-_L(v)) \\
			    & \geqslant & |l_\rho(v) - In(v) |
\end{eqnarray*}
 \end{proof}

In the following, we denote by $\sigma_\mathcal{A}$ the ranking returned by \textsc{Inc-Degree}, and by $\sigma_\mathcal{O}$ the ranking returned by any optimal solution.

\begin{lemma}
\label{lem:approxtwo}
Let $\rho : V \rightarrow V$ be any ranking. The following holds:
$$\displaystyle\sum_{v \in V} |l_\rho(v) - In(v)| \geqslant \displaystyle\sum_{v \in V} | l_{\sigma_\mathcal{A}}(v) - In(v)|$$
\end{lemma}

\begin{proof}
Observe that this result is trivial if $\rho$ orders the vertices by increasing in-degrees. So assume this is not the case. This means that there exists $i \in V$ such that $In(u) > In(v)$ and $u = \rho^{-1}(i)$, $v = \rho^{-1}(i+1)$. We build a new ranking $\rho'$ from $\rho$ by swapping $u$ and $v$. We now prove the following (the result will follow by induction): 
$$\displaystyle\sum_{v \in V} ( |l_\rho(v) - In(v) | - |l_{\rho'}(v) - In(v) |) \geqslant 0$$
To see this, notice that for every vertex $w \notin \{u,v\}$, we have $\rho(w) = \rho'(w)$, and hence the only terms non equal to $0$ are obtained in $u$ and $v$. This gives :
$$|l_{\rho}(u) - In(u) | - |l_{\rho'}(u) - In(u) | + |l_\rho(v) - In(v) | - |l_\rho'(v) - In(v) |$$
There are now three main cases to study. Observe that we have $\rho(v) = i + 1$, $\rho(u) = i$, $\rho'(v) = i$ et $\rho'(u) = i + 1$, and hence $l_\rho(u) = l_{\rho'}(v)$, $l_\rho(v) = l_{\rho'}(u)$ and $l_\rho(v) > l_\rho(u)$. 

\begin{enumerate}[(i)]
	\item $i + 1 < r$ : in this case, we obtain $0$ since $l_\rho(u) = l_{\rho'}(u) = 0$.
	\item $i + 1 = r$ : here, we have $ |-In(u)| - |1 - In(u)| + |1 - In(v)| - |-In(v)|$ = $(In(u) - In(v)) + (|1 - In(v)| - |1 - In(u)|)$. One can see that such a term is equal to $0$ if $In(u) > In(v) \geqslant 1$, and to $2$ if $In(u) > In(v) = 0$.
	\item $i + 1 > r$ : there we have $|l_\rho(u) - In(u)| - |l_{\rho'}(u) - In(u)| + |l_\rho(v) - In(v)| - |l_{\rho'}(v) - In(v)|$, which in turn is equal to :

\begin{eqnarray*}
	2\cdot \Big(min\big\{l_{\rho'}(v), In(v)\big\} - min\big\{l_\rho(u),In(u)\big\}\Big) \\ 
	+\ 2\cdot \Big(min\big\{l_{\rho'}(u),In(u)\big\} - min\big\{l_\rho(v),In(v)\big\}\Big)
\end{eqnarray*}

We now have to study several cases to be able to conclude (observe that $l_\rho(u) = l_{\rho'}(v)$ and $l_{\rho'}(u) = l_\rho(v)$):
\begin{enumerate}[(1)]
	\item $ l_\rho(u)  \geqslant In(u)$ : then the term is equal to $0$.
	\item $In(v) \geqslant  l_\rho(u)$ : then the term is equal to $2(min\{l_{\rho'}(u),In(u)\} - min\{l_\rho(v),$ $In(v)\})$. If $In(u) \leqslant l_{\rho'}(u)$, the term is equal to $2(In(u) - In(v))$, which is positive. If $In(u) > In(v) > l_\rho(v)$, the whole sum is again equal to $0$. Finally, if $In(v) \leqslant l_{\rho'}(u) < In(u)$, then the term is $2(l_{\rho'}(u) - In(v))$ which is again positive.
	\item $In(u) > l_\rho(u) > In(v)$ : the first term is $2(In(v) - l_\rho(u))$. If $In(u) > l_{\rho'}(u)$, then the sum is $2(In(v) - l_\rho(u)) + 2(l_{\rho'}(u) - In(v))$ which is non negative. Otherwise, the sum is $2(In(v) - l_\rho(u)) - 2(In(u) - In(v))$, which is equal to $2(In(u) - l_\rho(u))$, positive by hypothesis.
\end{enumerate}
\end{enumerate}
 \end{proof}

\begin{lemma}
\label{lem:approxthree}
Let $\rho, \gamma : V \rightarrow V$ be two rankings. The following holds:
$$\displaystyle{\displaystyle\sum_{v \in V}} |l_\rho(v) - l_\gamma(v)| \geqslant |b_\rho - b_\gamma|$$
\end{lemma}

\begin{proof}
We consider the set of constraints which are inconsistent in $\rho$ but not in $\gamma$ and denote it by $\mathcal{B}_{\rho \setminus \gamma}$. Similarly, we consider $\mathcal{B}_{\gamma \setminus \rho}$. First observe that:
$$\displaystyle\sum_{S \in \mathcal{B}_{\rho \setminus \gamma}} 1 + \displaystyle\sum_{S \in \mathcal{B}_{\gamma \setminus \rho}} 1 \geqslant |b_\gamma - b_\rho|$$

Moreover, we know by definition that $K(\rho, \gamma) \geqslant \displaystyle\sum_{S \in \mathcal{B}_{\rho \setminus \gamma}} 1 + \displaystyle\sum_{S \in \mathcal{B}_{\gamma \setminus \rho}} 1$. We need the following result. 

\begin{claim}
\label{claim:tau}
	The following holds: $\displaystyle{\displaystyle\sum_{v \in V}} |l_\rho(v) - l_\gamma(v)| \geqslant \mathcal{K}(\rho, \gamma) $. 

\end{claim}

\begin{proofclaim}
	Let $v \in V$ be any vertex. First, observe that since we are considering a dense instance, we have (assuming w.l.o.g. that $\rho(v) > \gamma(v)$): 
	
	$$
		\begin{array}{ll}
		|l_\rho(v) - l_\gamma(v)| & = \displaystyle { \rho(v) - 1 \choose r - 1 } -  { \gamma(v) -1 \choose r - 1 }  \\ 
		& \\
								  & = \displaystyle\sum_{i=0}^{r-2}  { \gamma(v) - 1 \choose i }  \cdot { \rho(v) - \gamma(v) \choose (r - 1) - i } 
								  \end{array}
	$$
	
	We now count the number of constraints whose consistency may have been modified by the change of position of $v$ in $\rho$ and $\gamma$. Let $S'$ be the set $\{v' \in V : \rho(v') < \rho(v)\}$. Moreover, we define the following sets: 
	
	\begin{minipage}{0.4 \linewidth}
	$$
		\begin{array}{ll}
		S_1 & = \{ s' \in S' : \gamma(s') < \gamma(v) \} \\
		S_2 & = \{ s' \in S' : \gamma(v) \leqslant \gamma(s') \leqslant \rho(v) \} \\
		S_3 & = S' \setminus (S_1 \cup S_2)
	\end{array}
	$$
	\end{minipage}
	\begin{minipage}{0.6 \linewidth}
		\centerline{\includegraphics[scale=1.75]{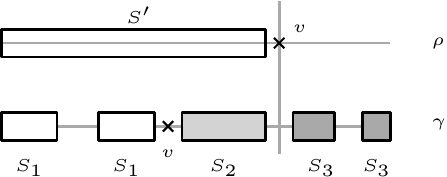}}
		\captionof{figure}{$S_1$, $S_2$ and $S_3$ defined according to $v \in V$.}
		\label{fig:kendall}
	\end{minipage}
	\vspace{0.4cm}
	
	Let $\mathcal{S}_v$ be the set of constraints of $\subseteq S_1 \cup S_2 \cup \{v\}$ whose consistency have been modified between the two rankings due to the change of position of $v$. Observe that any constraint $S$ in $\mathcal{S}_v$ must contain $v$. Moreover, $S£$cannot contain $r - 1$ vertices in $S_1$, since otherwise its consistency is not modified by moving $v$ (which would be the leftmost vertex of $S$ in both $\rho$ and $\gamma$). Since we are on a dense instance, we know that: 
	
	$$
		\begin{array}{ll}
		|\mathcal{S}_v| & \leqslant \displaystyle\sum_{i=0}^{r-2}  { \gamma(v) - 1 \choose i }  \cdot { \rho(v) - \gamma(v) \choose (r - 1) - i } \\
			& \\
						& \leqslant |l_\rho(v) - l_\gamma(v)|
		\end{array}
	$$
	
	A similar counting argument can be applied when $\rho(v) < \gamma(v)$ (observe that the consistency of any constraint is not modified by $v$ when $\rho(v) = \gamma(v)$). Altogether, we thus have: 
	
	$$
		\sum_{v \in V} |\mathcal{S}_v| \leqslant \displaystyle{\displaystyle\sum_{v \in V}} |l_\rho(v) - l_\gamma(v)|
	$$
	
	To conclude the proof, we show that $\sum_{v \in V} |\mathcal{S}_v| \geqslant \mathcal{K}(\rho, \gamma)$. To see this, notice that any constraint that belongs to $\mathcal{K}(\rho, \gamma)$ must contain two vertices $u$ and $v$ such that (w.l.o.g.) $u <_\rho v$ and $v <_\gamma u$, with the additional property that $S \in L_{\rho}(v)$ and $S \in L_\gamma(u)$. Indeed, by definition of \rFAST{}, the leftmost vertex of $S$ in $\rho$ cannot be the leftmost vertex of $S$ in $\gamma$, since otherwise we would have $c(\rho(S)) = c(\gamma(S))$. In particular, this means that $S$ will be counted in (at least) $\mathcal{S}_{max\{\rho(v), \gamma(u)\}}$. 
\end{proofclaim}

This concludes the proof of Lemma~\ref{lem:approxthree}.
 \end{proof}

We are now ready to prove the main result of this section :\\

\emph{Proof of Theorem~\ref{thm:approx}.} By the previous Lemmata, we have the following :
			$$
				\begin{array}{ll}
				\displaystyle 4b_{\sigma_\mathcal{O}} & \geqslant \sum_{v \in V} | l_{\sigma_\mathcal{O}}(v) - In(v) | + \sum_{v \in V} | l_{\sigma_\mathcal{O}}(v) - In(v) | \\
													& \geqslant \sum_{v \in V} | l_{\sigma_\mathcal{O}}(v) - In(v) | + \sum_{v \in V} | l_{\sigma_\mathcal{A}}(v) - In(v) | \\
													&  =  \sum_{v \in V} (| l_{\sigma_\mathcal{O}}(v) - In(v) | + | l_{\sigma_\mathcal{A}}(v) - In(v) |) \\
													& \geqslant \sum_{v \in V} | l_{\sigma_\mathcal{O}}(v) - l_{\sigma_\mathcal{A}}(v) | \\
													& \geqslant b_{\sigma_\mathcal{A}} - b_{\sigma_\mathcal{O}}
			\end{array}
			$$
			Hence we have $b_{\sigma_\mathcal{A}} \leqslant 5b_{\sigma_\mathcal{O}}$, which implies the result. 
\hfill \qed

\subsubsection{Kernelization algorithm}

In order to design our kernelization algorithm for this problem, we need to study the topology of conflicts that contain exactly one inconsistent constraint. As we shall see, the configuration for \FASHT{} is slightly different to the ones previously observed. In particular, as we shall see, the problem is not $l_r$-simply characterized. However, the addition of a new reduction rule will allow us to conclude as in the other cases.  

\begin{lemma}
\label{lem:fashtconflict}
	Let $R_\sigma = (V, c, \sigma)$ be an ordered instance of {\sc $r$-FAST}, and $C = \{s_1, \ldots, s_{r+1}\}$ be a set of $r + 1$ vertices such that $s_i <_\sigma s_{i+1}$ for every $1 \leqslant i \leqslant r$. Assume $R_\sigma[C]$ contains exactly one inconsistent constraint $S$. Then $C$ is a conflict iff $S$ is unconsecutive or $S = \{s_2, \ldots, s_{r+1}\}$.
\end{lemma}

\begin{proof}
	Assume first that $S$ is unconsecutive (\ie there exists $2 \leqslant i \leqslant r$ such that $S = C \setminus \{s_i\}$). Since any constraint different from $S$ is consistent w.r.t. $\sigma$, it follows that $sel(C \setminus \{s_j\}) = s_{r+1}$ for any $1 \leqslant j \leqslant r,\ j \neq \{l,i\}$. Together with the fact that $sel(S) = s_l$ for $l < r+1$, this implies that any consistent ranking $\rho$ must verify $s_l <_\rho s_{r+1}$ and $s_{r+1} <_\rho s_l$ which cannot be.
	We now deal with the case where $S$ is induced by $\{s_2, \ldots, s_{r+1}\}$. Once again, this means that $sel(S) = s_l$, $2 \leqslant l < r +1$, and $sel(C \setminus \{s_j\}) = s_{r+1}$ for any $2 \leqslant j \leqslant r,\ j \neq l$, which is inconsistent with any ranking. 
	Finally, assume that $S$ is induced by $\{s_1, \ldots, s_r\}$. Now, by swapping  
	the vertices $sel(S)$ and $s_r$ (notice that there is no constraint where $s_r$ is 
	the selected vertex), we 
	obtain a consistent ranking, implying that $C$ is not a conflict in this case. 
 \end{proof}
	
\begin{corollary}
\label{coro:fashtlr}
	There does not exist $l_r \in \mathbb{N}$ such that {\sc $r$-FAST} is $l_r$-simply characterized. 
\end{corollary}

\begin{proof}
	Let $R = (V,c)$ be an instance of {\sc $r$-FAST} and $q \in \mathbb{N}$, $q > r$. 
	Let $C = \{s_1, \ldots, s_q\}$ be any set of vertices ordered 
under some ranking $\sigma$ such 
that $s_i <_\sigma s_{i+1}$ for $1 \leqslant i < q$. Assume that $S = \{s_1, \ldots, s_r\}$ is the only inconsistent constraint of $R_\sigma[C]$. By Lemma~\ref{lem:fashtconflict}, we know that $C$ is not a conflict, implying that $r$-FAST is not $q$-simply 
characterized. 
 \end{proof}

We need the following rule, 
which implies that the last vertex of any ordered instance of {\sc $r$-FAST} belongs 
to an inconsistent constraint. 

\begin{polyrule}
\label{rule:uselessvertex}
	Let $v$ be any vertex which is selected in every constraint containing it. Remove 
	$v$ and any constraint containing it from $R = (V,c)$. 
\end{polyrule}

\begin{lemma}
\label{lem:uselessvertex}
	Rule~\ref{rule:uselessvertex} is sound and can be applied in polynomial time.
\end{lemma}

\begin{proof}
	First, observe that any edition of size at most $k$ for the original instance will yield an edition for the reduced one. In the other direction, assume that $R_v = (V \setminus \{v\}, c)$ admits an edition $\mathcal{F}$ of size at most $k$, and let $\sigma$ be the consistent ranking obtained after editing the constraints of $\mathcal{F}$. Since adding $v$ to the end of $\sigma$ does not introduce any inconsistent constraint, $\mathcal{F}$ is also an edition for the original instance.
 \end{proof}

Observe that a given instance can contain at most one such vertex.~We thus iteratively apply this rule until no vertex selected in every constraint containing it remains. 
Given any constraint $S$ of an ordered instance $R_\sigma = (V, c, \sigma)$, $span^-(S)$ denotes the set containing $span(S)$ and all vertices lying before $S$ in $\sigma$.~Observe that by Lemma~\ref{lem:fashtconflict}, any set $C \subseteq V$ of $r+1$ vertices such that $R_\sigma[C]$ contains exactly one inconsistent constraint is a conflict iff $C \subseteq span^-(S)$. 

\begin{polyrule}
\label{rule:correctFASHT}
	Let $R_\sigma = (V, c, \sigma)$ be an ordered instance of {\sc $r$-FAST}  and $\mathcal{S} = \{C_1, \ldots, C_m\}$, $m > k$, 
	be a simple sunflower of center $S$. Edit $S$ w.r.t. $\sigma$ and decrease $k$ by $1$.

\end{polyrule}

\begin{lemma}
\label{lem:correctFASHT}
	Rule~\ref{rule:correctFASHT} is sound. 

\end{lemma}

\begin{proof}
	Let $\mathcal{F}$ be any edition of size at most $k$: by Lemma~\ref{lem:flower}, $\mathcal{F}$ must contain $S$. Since $|\mathcal{F}| \leqslant k$, there exists $1 \leqslant i \leqslant m$ such that $S$ is the only constraint edited by $\mathcal{F}$ in $R[C_i]$. Assume that $S$ was not edited w.r.t. $\sigma$: since no other constraint has been edited in $R[C_{i}]$, $R_\sigma[C_i]$ still contains exactly one inconsistent constraint (namely $S$). Observe now that, by definition of a simple conflict, $\cup_{i=1}^m C_i \subseteq span^-(S)$ holds for any simple sunflower $\mathcal{S}$. Hence 
	Lemma~\ref{lem:fashtconflict} implies that $C_i$ is 
	a conflict, contradicting the fact that $\mathcal{F}$ is an edition. 
 \end{proof}

Here again, we can prove the existence of a simple sunflower under a certain cardinality assumption. 

\begin{lemma}
	\label{lem:correctrFAST}
	Let $R_\sigma = (V, c, \sigma)$ be an ordered instance of {\sc $r$-FAST}  with at most $p \geqslant 1$ inconsistent constraints, and $S$ be an inconsistent constraint with $|span^-(S)| > p + k + r$. Then $S$ is the center of a simple sunflower $\{C_1, \ldots, C_m\}$, $m > k$.
\end{lemma}

\begin{proof}
Since there at most $p$ vertices $v$ such that $R_\sigma[S \cup \{v\}]$ contains more than one inconsistent constraint, there exist $k + 1$ vertices $\{s_1, \ldots, s_{k+1}\} \subseteq span^-(S)$ such that $C_i = S \cup \{s_i\}$ contains exactly one inconsistent constraint, $1 \leqslant i \leqslant k+1$. Since $s_i \in span^-(S)$ for $1 \leqslant i \leqslant k +1$, Lemma~\ref{lem:fashtconflict} implies that $C_i$ is a conflict for $1 \leqslant i \leqslant k+1$. It follows that $S$ is the center of the simple sunflower $\{C_1, \ldots, C_{k+1}\}$.
 \end{proof}

\begin{theorem}
\label{thm:fasht}
	{\sc $r$-FAST} admits a kernel with at most $6k + r$ vertices.
\end{theorem}

\begin{proof}
	Let $R = (V, c)$ be an instance of {\sc $r$-FAST} . We start by running the constant-factor approximation (Theorem~\ref{thm:approx}) on $R$, obtaining a ranking $\sigma$ of $V$ with at most $p$ inconsistent constraints. Notice that we may assume $p > k$ and $p \leqslant 5k$, since otherwise we return a small trivial \textsc{Yes}- (resp. \textsc{No}-)instance. Assume that $|V| > p + k + r$ and let $S$ be any inconsistent constraint containing $v_n$ (thus $span^-(S) = V$). Observe that $S$ is well-defined since $R$ is reduced by Rule~\ref{rule:uselessvertex} (and hence $v_n$ must belong to an inconsistent constraint). 
	By Lemma~\ref{lem:correctrFAST}, it follows that we can find a simple sunflower $\{C_1, \ldots, C_m\}$, $m > k$ in polynomial time. We thus apply Rule~\ref{rule:correctFASHT} and edit $S$ w.r.t. $\sigma$. We now apply Rule~\ref{rule:uselessvertex} and repeat this process until we either do not find a large enough simple sunflower or $k < 0$. In the former case, Lemma~\ref{lem:correctrFAST} implies that $|V| \leqslant p + k + r \leqslant 6k + r$, while in the latter case we return a small trivial \textsc{No}-instance.  
 \end{proof}

\section{Conclusion}

In this paper, we considered the kernelization of several dense ranking $r$-CSPs. 
In particular, we proved that any ranking $r$-CSP that can be \emph{simply characterized} and that 
admits a constant factor-approximation algorithm also admits a linear vertex-kernel. Interestingly, our kernelization algorithm makes use of a modification of the classical 
sunflower rule, which usually provides \emph{polynomial} kernels.  
As a main consequence of this result, we improved the size of the kernel for \BIT{}, and described the first polynomial kernels for {\sc $r$-BIT} and {\sc $r$-TFAST}. 
A natural question is thus whether these techniques be applied 
to dense ranking $r$-CSP for (weakly-)fragile constraints~\cite{KS11}? Such problems are known to be fixed-
parameter tractable~\cite{KS10}, but the existence of a polynomial kernel remains an 
open problem. We also considered another generalization of \FAST{}, namely \FASHT{}. We obtained a $5$-approximation algorithm as well as a linear vertex-kernel for this problem. Investigating analogy with \FAST{}, it would be interesting to determine whether this problem admits a PTAS. \\

{\small

\noindent \textbf{Acknowledgments.} Research supported by the AGAPE project  (ANR-09-BLAN-0159). The author would like to thank Christophe Paul, St\'ephan Thomass\'e,  Mathieu Liedloff and Mathieu Chapelle for helpful discussions and comments. 

\bibliographystyle{plain}
\bibliography{ranking}

\begin{thebibliography}{10}

\bibitem{Abu10}
F.~N. Abu-Khzam.
\newblock A kernelization algorithm for d-hitting set.
\newblock {\em J. Comput. Syst. Sci}, 76(7):524--531, 2010.

\bibitem{AA07}
N.~Ailon and N.~Alon.
\newblock Hardness of fully dense problems.
\newblock {\em Inf. Comput}, 205(8):1117--1129, 2007.

\bibitem{Alo06}
N.~Alon.
\newblock Ranking tournaments.
\newblock {\em SIAM J. Discrete Math}, 20(1):137--142, 2006.

\bibitem{ALS09}
N.~Alon, D.~Lokshtanov, and S.~Saurabh.
\newblock Fast fast.
\newblock In {\em ICALP}, pages 49--58, 2009.

\bibitem{BP11}
S.~Bessy and A.~Perez.
\newblock Polynomial kernels for proper interval completion and a related
  problem.
\newblock In {\em FCT}, volume 6914, pages 1732--1744, 2011.

\bibitem{Bod09}
H.~L. Bodlaender.
\newblock Kernelization: New upper and lower bound techniques.
\newblock In {\em IWPEC}, pages 17--37, 2009.

\bibitem{BDFH09}
H.~L. Bodlaender, R.~G. Downey, M.~R. Fellows, and D.~Hermelin.
\newblock On problems without polynomial kernels.
\newblock {\em JCSS}, 75(8):423--434, 2009.

\bibitem{BJK11}
H.~L. Bodlaender, B.~M.~P. Jansen, and S.~Kratsch.
\newblock Cross-composition: {A} new technique for kernelization lower bounds.
\newblock In {\em STACS}, volume~9 of {\em LIPIcs}, pages 165--176, 2011.

\bibitem{CTY07}
P.~Charbit, S.~Thomass{\'e}, and A.~Yeo.
\newblock The minimum feedback arc set problem is {NP}-hard for tournaments.
\newblock {\em Combinatorics, Probability \& Computing}, 16(1):1--4, 2007.

\bibitem{CFR06}
D.~Coppersmith, L.~Fleischer, and A.~Rudra.
\newblock Ordering by weighted number of wins gives a good ranking for weighted
  tournaments.
\newblock In {\em SODA}, pages 776--782, 2006.

\bibitem{DF99}
R.G. Downey and M.~R. Fellows.
\newblock {\em Parameterized Complexity}.
\newblock Springer, 1999.

\bibitem{GHPP12}
S.~Guillemot, F.~Havet, C.~Paul, and A.~Perez.
\newblock On the (non-)existence of polynomial kernels for $p_l$ -free edge
  modification problems.
\newblock In {\em Algorithmica}, 2012.
\newblock to appear.

\bibitem{KS10}
M.~Karpinski and W.~Schudy.
\newblock Faster algorithms for feedback arc set tournament, kemeny rank
  aggregation and betweenness tournament.
\newblock In {\em ISAAC}, volume 6506 of {\em LNCS}, pages 3--14, 2010.

\bibitem{KS11}
M.~Karpinski and W.~Schudy.
\newblock Approximation schemes for the betweenness problem in tournaments and
  related ranking problems.
\newblock In {\em APPROX-RANDOM}, volume 6845 of {\em LNCS}, pages 277--288,
  2011.

\bibitem{KW09}
S.~Kratsch and M.~Wahlstr{\"o}m.
\newblock Two edge modification problems without polynomial kernels.
\newblock In {\em IWPEC}, volume 5917 of {\em LNCS}, pages 264--275, 2009.

\bibitem{Nie06}
R.~Niedermeier.
\newblock {\em Invitation to Fixed Parameter Algorithms (Oxford Lecture Series
  in Mathematics and Its Applications)}.
\newblock {Oxford University Press, USA}, March 2006.

\bibitem{PPT11}
C.~Paul, A.~Perez, and S.~Thomass\'e.
\newblock Conflict packing yields linear-vertex kernels for $k$-{FAST},
  $k$-{DENSE RTI} and a related problem.
\newblock In {\em MFCS}, volume 6907 of {\em LNCS}, pages 497--507, 2011.

\bibitem{Tho10}
S.~Thomass{\'e}.
\newblock A $4k^2$ kernel for feedback vertex set.
\newblock {\em ACM Transactions on Algorithms}, 6(2), 2010.

\bibitem{vBMN10}
R.~van Bevern, H.~Moser, and R.~Niedermeier.
\newblock Kernelization through tidying.
\newblock In {\em LATIN}, volume 6034 of {\em LNCS}, pages 527--538, 2010.

\bibitem{WKNU09}
M.~Weller, C.~Komusiewicz, R.~Niedermeier, and J.~Uhlmann.
\newblock On making directed graphs transitive.
\newblock In {\em WADS}, volume 5664 of {\em LNCS}, pages 542--553, 2009.

\end{thebibliography}
}

\end{document}